\documentclass{article}

\usepackage{balance} 

\usepackage{graphicx} 
\usepackage{complexity,hyperref,url}
\usepackage{xcolor}
\usepackage{todonotes}
\usepackage{amsthm,amsmath}
\usepackage[capitalise]{cleveref}
\usepackage{xspace}
\usepackage[ruled]{algorithm2e}
\usepackage{array}
\usepackage{vmargin}
\usepackage{authblk}
\usepackage{amsfonts}
\usepackage{natbib}

\newcommand{\qedclaim}{\hfill $\diamond$ \medskip}

\newtheorem{theorem}{Theorem}
\newtheorem{corollary}{Corollary}
\newtheorem{lemma}{Lemma}
\newcommand{\eps}{\varepsilon}

\newcommand{\XNLP}{\mathsf{XNLP}}

\newcommand{\EI}{{\normalfont\textsc{Eli}\-\textsc{min}\-\textsc{inating} \textsc{Il}\-\textsc{lu}\-\textsc{sion}}\xspace}
\newcommand{\TVD}{{\normalfont\textsc{Total} \textsc{Vector} \textsc{Domina}\-\textsc{tion}}\xspace}

\newcommand{\elil}{{\textsc{Eli}\-\textsc{min}\-\textsc{inating} \textsc{Il}\-\textsc{lu}\-\textsc{sion}}\xspace}
\newcommand{\selil}{{\normalfont\textsc{Sub}\-\textsc{set}\ \textsc{Eli}\-\textsc{min}\-\textsc{inat}\-\textsc{ing} \textsc{Il}\-\textsc{lu}\-\textsc{sion}}\xspace}
\newcommand{\elilShort}{\textsc{EI}\xspace}

\newcommand{\defproblem}[3]{
    \vspace{3mm}
    \noindent\fbox{
        \begin{minipage}{0.97\textwidth}
            #1\newline
            \textbf{Input:} #2\\
            \textbf{Task:} #3
        \end{minipage}
    }
    \vspace{3mm}
}

\newcommand\restr[2]{
  \left.\kern-\nulldelimiterspace
  #1
  \right|_{#2}
}

\title{Eliminating Majority Illusions\thanks{This work will be presented in AAMAS'25}}

\author[1]{Foivos Fioravantes}
\author[2]{Abhiruk Lahiri}
\author[3]{Antonio Lauerbach}
\author[4]{Lluís Sabater}
\author[3]{Marie Diana Sieper}
\author[3]{Samuel Wolf}

\affil[1]{Department of Theoretical Computer Science, FIT, Czech Technical University in Prague, Czech Republic}
\affil[2]{Heinrich Heine University, D\"{u}sseldorf, Germany}
\affil[3]{University of W\"{u}rzburg, Germany}
\affil[4]{Charles University, Prague, Czech Republic}

\date{}

\newcommand{\BibTeX}{\rm B\kern-.05em{\sc i\kern-.025em b}\kern-.08em\TeX}

\begin{document}



\maketitle

\begin{abstract}
An opinion illusion refers to a phenomenon in social networks where agents may witness distributions of opinions among their neighbours that do not accurately reflect the true distribution of opinions in the population as a whole. A specific case of this occurs when there are only two possible choices, such as whether to receive the COVID-19 vaccine or vote on EU membership, which is commonly referred to as a \textit{majority illusion}. In this work, we study the topological properties of social networks that lead to opinion illusions and focus on minimizing the number of agents that need to be influenced to eliminate these illusions. To do so, we propose an initial, but systematic study of the algorithmic behaviour of this problem.

We show that the problem is $\NP$-hard even for underlying topologies that are rather restrictive, being planar and of bounded diameter.
We then look for exact algorithms that scale well as the input grows ($\FPT$). We argue the in-existence of such algorithms even when the number of vertices that must be influenced is bounded, or when the social network is arranged in a ``path-like'' fashion (has bounded pathwidth). On the positive side, we present an $\FPT$ algorithm for networks with ``star-like'' structure (bounded vertex cover number). Finally, we construct an $\FPT$ algorithm for ``tree-like'' networks (bounded treewidth) when the number of vertices that must be influenced is bounded. This algorithm is then used to provide a $\PTAS$ for planar graphs.
\end{abstract}

\section{Introduction}
\emph{Opinion illusion} occurs for an agent in a social network when the perception within its immediate neighbourhood differs from the broader network's predominant opinion.
In democratic societies, decisions often hinge on binary choices such as EU membership or the rightfulness of COVID-19 vaccination.
Opinion illusion in situations with binary choices is called a \emph{majority illusion}.
Over time, an agent under illusion may alter its opinion if left unaddressed.
This change further influences other agents to reconsider their opinions.
This cascading effect leads to the spread of misinformation and bias in the community, an undesirable scenario.
An exemplary visualisation of the problem is presented in Figure~\ref{fig:example}.

\begin{figure}[h]
    \centering
    \includegraphics[width=0.7\textwidth]{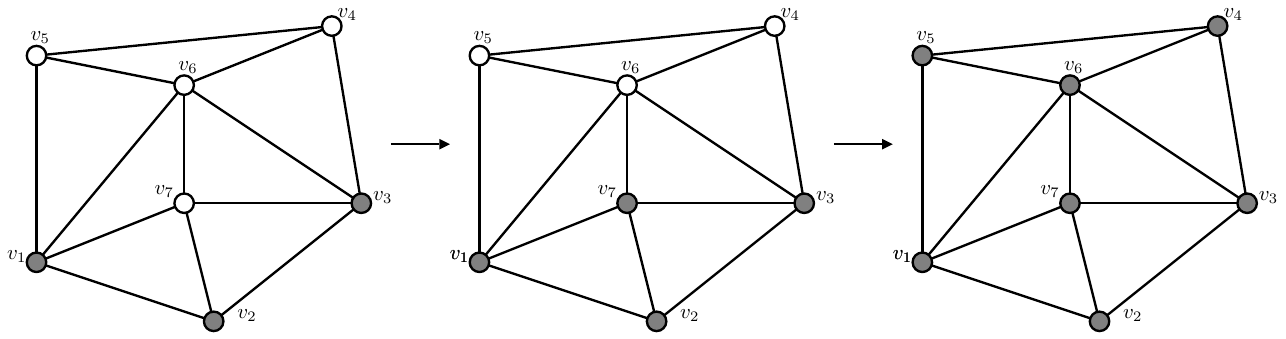}
    \caption{In the leftmost figure, the vertex $v_7$ is under illusion. It may change its colour over time and subsequently influence the colour of $v_6$ followed by $v_5$ and $v_4$. The final configuration is presented in the rightmost figure.}
    \label{fig:example}
\end{figure}

The process of false information spreading may create significant differences in the outcome of an election.
For example, the outcome of the 2018 French and 2018 Italian elections was aligned with the information/misinformation spread from bots~\cite{Ferrara17, AlaphilippeCCM18}.
\cite{CastiglioniF0L21} did a theoretical study of the problem of manipulating elections through social influence.
\cite{WEKV22} showed that although controlling election is generally $\NP$-hard, agents can be partitioned into similar groups and the problem becomes tractable.
\cite{FaliszewskiGKT22} further showed $\FPT$ algorithms to influence the agents in a social network to obtain desirable election outcomes.

The spread of opinion in social networks and arriving at a consensus has been extensively studied earlier~\cite{AulettaFG20, BredereckE17, FaliszewskiGKT22}.
\cite{AulettaFG20} studied the complexity questions of finding a minimum-sized set of agents to maximise the spread of information.
This problem is known under the name of \emph{target set selection}~\cite{KempeKT05}.
\cite{BredereckE17}considered the problem of maximising opinion diffusion under majority settings with the help of bribery (or influence).
They further established a connection between their model and the target set selection problem.
To counter misinformation spreading, a tried and tested approach involves network intervention~\cite{Valente12}.
This is often achieved by appointing \emph{champions}, typically opinion leaders or influencers~\cite{LatkinK12}.
This method has effectively driven change, notably in health behaviour modification~\cite{ValenteP12,StarkeyAHMC09}.

\cite{LermanYW16} first considered majority illusion, examining social networks where such phenomena occur.
A graph-theoretic study of different variants of majority illusion is reported in ~\cite{VenemaLosCG23}.
\cite{GrandiKLST23} initiated the algorithmic study of the problem of verifying the existence of and eliminating majority illusion.
They considered the problem where along with the social network a parameter $q \in [0,1]$ is given, which denotes at least $q$-fraction of the vertices under an illusion in the social network.
They showed the decision problem, that there is a labelling which induces a $q$-majority illusion, is $\NP$-complete for every rational number $q \in (\frac{1}{2}, 1]$ even for planar graphs, bipartite graphs as well as graphs of bounded maximum degree.
On the positive side, they showed $\FPT$ algorithms when parameterised by neighbourhood diversity, vertex cover, maximum degree with treewidth or cliquewidth.
Furthermore, they proposed two editing operations, edge addition and edge deletion to eliminate majority illusion. From a theoretical point of view, it is rather natural to consider these two operations. However, one should also consider their associated cost, meant to model the effort required to create or remove connections between strangers or friends respectively.

We study the theoretical model of the majority illusion problem while considering the practicality of the proposed solution.
A theoretical study of eliminating undesirable properties within social networks is not new. It has been extensively explored in recent years~\cite{BhawalkarKLRS15,ChitnisFG13}.
We consider the network intervention method as suggested in~\cite{Valente12,LatkinK12}, where the objective is to find the minimum number of leaders or influencers.
The task of the leader or influencer is to sync its opinion with the global majority opinion\footnote{A similar solution has recently been proposed independently by \cite{ChitnisP}.} and by doing so we remove all illusions of the network.
Finding a smallest possible set of nodes/agents in a social network to create desirable influence is a well-established research interest and has been studied in the context of digital marketing \cite{DinhNT12, DinhZNT14, HeJBC14, PanB19}.

\paragraph*{Related problem of independent theoretical interest}
We begin this work by observing a strong relation (polynomial equivalence) between the \elil (\elilShort for short) problem and a variation of the classic dominating set problem, known as \textsc{Total Vector Domination} (\textsc{TVD} for short). To the best of our knowledge, the \textsc{TVD} problem has only been considered in~\cite{IshiiOU16,CicaleseMV13}. We would like to stress here that none of our infeasibility results follow directly from the known results for the \textsc{TVD} problem. On the contrary, our reduction serves as a way to translate many of the efficient algorithm that can be conceived for \elilShort, including the $\PTAS$ we provide, into their direct counterparts for the \textsc{TVD} problem. This last observation has particular importance in view of the sparsity of positive results that exist for the \textsc{TVD} problem.

\paragraph*{Our Contribution.}
We begin by showing the aforementioned equivalence between the \elilShort and \textsc{TVD} problems. We then focus on the \elilShort problem. We show that this problem is $\NP$-hard even on planar bipartite graphs. Moreover, the problem is $\W[2]$-hard when parameterised by the solution size, i.e., the minimum number of vertices that must be influenced. Both of the previous results hold even if we restrict the input graph to have bounded diameter. It is then natural to wonder about a possible efficient algorithm for solving the problem when considering structural parameters of the input graph. Unfortunately, we show that the problem is $\XNLP$-hard when parameterised by the pathwidth of the input graph (implying the same result when parameterised by the treewidth). Nevertheless, we do provide an \FPT~algorithm parameterised by the vertex cover number of the input graph. Finally, we provide a $\PTAS$ for planar graphs. To achieve this we also construct an $\FPT$~algorithm parameterised by the treewidth of the input graph plus the solution size. This implies an $\XP$ algorithm parameterised just by the treewidth, which is then used as a building block for the aforementioned $\PTAS$.

\section{Preliminaries}

Formally, we consider social networks as graphs $G = (V, E)$ where each vertex has a labelling $f: V(G) \to \{0,1\}$.
We assume there are strictly more vertices with label $0$ than $1$ in $G$.
We say that a vertex is under illusion if the label $1$ has a surplus in its neighbourhood. That is, a vertex $v\in V$ \emph{is under illusion by $f$} if $|\{u\in N(v):f(u)=1\}|>|\{u\in N(v):f(u)=0\}|$.
We say that a labelling $f'$ is a \emph{solution} to the majority illusion problem on $G$ if $f'$ induces no illusion; the \emph{size} of this solution is $|\{v\in V(G):f(v)\neq f'(v)\}|$.

\defproblem{\elil}
{A graph $G= (V,E)$, a labelling $f: V\to \{0,1\}$ and an integer $k\geq 1$.}
{Is there a labelling $f': V\rightarrow \{0,1\}$ such that
$f'$ induces no illusion and
$|\{v\in V:f(v)\neq f'(v)\}|\leq k$.
}

We will follow the usual graph theory notation~\cite{D12}. For any vertex $v$ of a graph $G$, we denote by $d_G(v)$ degree of $v$ in $G$, which is equal to the number of neighbours of $v$ in $G$. Whenever it is obvious from the context, the subscript will be dropped.
\paragraph*{Parameterised Complexity.}
The goal in the field of parameterised complexity is to construct exact algorithms that are efficient with respect to a measure of time that is extended by a secondary measure of the problem, commonly referred to as the \emph{parameter}. Let $n$ denote the size of the input of a problem, $k$ denote the considered parameter and $f$ be an arbitrary computable function. We consider that a parameterised problem is solved efficiently if it can be determined in $f(k)\cdot n^{\mathcal{O}(1)}$ time. In such cases, we say that the problem is \emph{fixed-parameter tractable} and that it belongs to the class \FPT.
A parameterised problem is \emph{slicewise polynomial} if it can be determined in $n^{f(k)}$ time. In such cases, we say that the problem belongs to the class \XP.
It should be noted that, unlike \NP-complete problems, there is actually a whole hierarchy of infeasibility for parameterised problems, referred to as the \W~hierarchy. A problem is presumably not in \FPT{} if there exists a $t\geq 1$ such that the problem is \W[$t$]-hard (by a parameterised reduction). Moreover, it is hypothesised that \W[$t$] $\subseteq$ \W[$t+1$] for every $t\geq 1$. Finally, we will need the definition of the recently introduced $\XNLP$ class~\cite{BGNS24}. This class contains the parameterised problems whose input can be encoded with $n$ bits and can be solved non-deterministically in time $f(k)\cdot n^{\mathcal{O}(1)}$ and space $f(k) \log n$.
We refer the interested reader to now classical monographs~\cite{CyganFKLMPPS15,Niedermeier06,FlumG06,DowneyF13} for a more comprehensive introduction to the topic of parameterised complexity.

\paragraph*{Structural Parameters.}
Let $G=(V,E)$ be a graph. A set $U\subseteq V$ is a \emph{vertex cover} if for every edge $e\in E$ it holds that $U\cap e \not= \emptyset$. The \emph{vertex cover number} of $G$, denoted $\mathsf{vc}(G)$, is the minimum size of a vertex cover of $G$.

A \emph{tree decomposition} $\mathcal{T} = (T,$ $\{X_t\}_{t \in V(T)})$ of $G$ is a tree $T$, such that the following hold:

\begin{itemize}
    \item Every node $t \in V(T)$ has an associated bag $X_t \subseteq V$ such that the union of all bags is equal to $V(G)$.
    \item For each edge $\{u, v\} \in E(G)$, there has to exist at least one bag $X_t$ with $u, v \in X_t$.
    \item For each vertex $v \in V(G)$, the nodes whose bags contain $v$ induce a connected subtree of $T$.
\end{itemize}

The \emph{width} of a tree decomposition is $\max\{|X_t|\mid t \in V(T)\} -1$. The \emph{treewidth} $\operatorname{tw}(G)$ of a graph $G$ is the smallest value, such that there exists a tree decomposition of $G$ with this width.

It is known that computing a tree decomposition of minimum width is fixed-parameter tractable when parameterised by the treewidth~\cite{Kloks94,Bodlaender96}, and even more efficient algorithms exist for obtaining near-optimal tree decompositions~\cite{KorhonenL23}.

A tree decomposition $(T,$ $\{X_t\}_{t \in V(T)})$ is \emph{nice}~\cite{Bodlaender98} if $T$ is rooted in $r\in V(T)$ and every node $t\in V(T)$ is exactly of one of the following four types:

\begin{enumerate}
    \item \textbf{Leaf:} $t$ is a leaf of $T$ and $|X_t|=1$.

    \item \textbf{Introduce:} $t$ has a unique child $i$ and there exists $v\in V$ such that $X_t=X_{i}\cup \{v\}$.

    \item \textbf{Forget:} $t$ has a unique child $i$ and there exists $v\in V$ such that $X_{i}=X_t\cup \{v\}$.

    \item \textbf{Join:} $t$ has exactly two children $i,j$ and $X_t=X_i=X_j$.
\end{enumerate}
It is well known that every graph $G=(V,E)$ admits a nice tree decomposition rooted in $r\in V(T)$, that has width equal to $\mathsf{tw}(G)$, $|V(T)|=\mathcal{O}(|V|)$ and $X_r=\{\emptyset\}~$\cite{Bodlaender98}.

The notions of (nice) \emph{path decomposition} and \emph{pathwidth} are defined analogously, by replacing the third item in the definition of a tree decomposition by the following: for every vertex $v\in V$, the nodes whose bags contain $v$ induce a connected \emph{subpath} of $\mathcal{T}$. Finally, nice path decompositions do not contain any join nodes.

\paragraph*{Approximation}
The goal of an approximation algorithm is to obtain an approximated solution of an intractable problem in polynomial time. Formally speaking, given a minimisation problem $\mathcal{P}$, a polynomial time algorithm $A$ is an approximation algorithm with an approximation ratio $\alpha \in \mathbb{R}$ if for all instances $I \in \mathcal{P}$, $A$ produces a feasible solution $\mathsf{SOL}(I)$ such that $|\mathsf{SOL}(I)| \leq \alpha\cdot |\mathsf{OPT}(I)|$, where $\mathsf{OPT}(I)$ is the optimum solution of $I$.
A $\PTAS$ for a minimisation problem is an approximation algorithm which for every $\eps > 0$ outputs a solution of size $(1+\eps)|\mathsf{OPT}|$ in time polynomial in the size of the input.

\section{Connection to total vector domination}
In this section, we will establish a polynomial-time equivalence between the \textsc{TVD} and \EI problems. This equivalence allows us to, to some extent, interchange results and complexity analyses between the two problems.
We begin by formally stating the definition of the \textsc{TVD} problem:

\defproblem{\TVD}
{A graph $G = (V, E)$ and a vector $(k(v) : v \in V)$ where $k(v) \in \{0, 1, \dots, d(v)\}$ for all $v \in V$.}
{Find a minimum-size set $S \subseteq V$ such that $|S \cap N(v)| \geq k(v)$ for all $v \in V$.}

\subsection*{From \elilShort to \textsc{TVD}}
Given an \elilShort instance $(G, f)$, we will construct a \textsc{TVD} instance $(G', k)$. The main obstacle we have to overcome is that there are some vertices already labelled as $0$ in $(G, f)$ that should influence the \textsc{TVD} solution, but are not in the solution of \elilShort.
We construct the graph $G'= (V',E')$ in a specific way to combat this limitation. We start with a copy of $G=(V,E)$. Then, for each vertex $v\in V$ with $f(v) = 0$, we attach a leaf $w$ to the vertex $v$. This finishes the construction of $G'$. We then define $k(v)$ for all $v \in V'$ as follows:
\[
k(v) =
\begin{cases}
   \lceil \frac{d_{G}(v)}{2} \rceil & \mbox{if } v \in V \\
    1 & \mbox{if } v \in V' \setminus V
\end{cases}
\]
This construction ensures that vertices labelled $0$ in $G$ are selected in the \textsc{TVD} solution of $G'$, as we force the leaves of $G'$ to have $k(w) = 1$. Moreover, any \textsc{TVD} solution eliminates all the illusions in the original graph, as we set $k(v)$ to be at least half of the nodes for all vertices $v \in V$. It is important to note that the solution given by the \textsc{TVD} problem will always choose vertices of $G$ as they are at least as good as the new vertices added. This leads us to the following lemma.


\begin{lemma}\label{lem:tvd-ei}
A solution to the \textsc{TVD} problem on $(G',k)$ corresponds to a solution of the \elilShort problem on $(G,f)$, excluding vertices of $G$ originally labelled $0$.
\end{lemma}

Observe that $G'$ is essentially the same as $G$, with some additional leaves. Thus, we obtain the following corollaries that follow directly from results in~\cite{CicaleseCGMV14,IshiiOU16}:

\begin{corollary}\label{cor:trees}
There is a polynomial time algorithm for \elilShort on trees.
\end{corollary}

\begin{corollary}
There is an $\FPT$ algorithm for \elilShort on planar graphs.
\end{corollary}


\subsection*{From \textsc{TVD} to \elilShort}
We now present the reverse reduction, transforming a \textsc{TVD} instance into an equivalent \EI instance.

Given a \textsc{TVD} instance $(G, k)$, we construct an \elilShort instance $(G', f)$ as follows. We start with $G'=(V',E')$ being a copy of $G=(V,E)$, with one copy of the $P_2$ path attached to each $v \in V$. We say that the vertices that are common between $G$ and $G'$ (i.e., $v\in V\setminus V'$) are the \emph{original} vertices, while any other vertex is \emph{auxiliary}. We label all the original vertices of $G'$ by $1$ and its auxiliary vertices by $0$. 
Then, for each original vertex $v$ of $G'$, we define the \emph{budget} of $v$ to be the value $B(v) = 2k(v) + 1 - d_{G}(v)$. Observe that a vertex might have a negative budget. For each original vertex $v\in V'$ such that $B(v)\geq 0$, we add $B(v)$ copies of the $P_3$ path to $G'$, attached to $v$, labelled by the labels $1,0,0$, with the newly added vertex labelled $1$ being the neighbour of $v$. Then, for each original vertex $v\in V'$ such that $B(v)<0$, we add $|B(v)|$ copies of the $P_2$ path to $G'$, attached to $v$ and labelled by the labels $0,0$. This finishes the construction of the instance $(G',k)$.

Notice that the above construction ensures that label $0$ is indeed the majority. Moreover, we may assume that any optimal solution of \elilShort on $(G',f)$ only relabels a subset of the original vertices of $G'$. Indeed, if an optimal solution contains an auxiliary vertex $u$ belonging to some path attached to an original vertex $v$, then it suffices to relabel any original neighbour of $v$ instead of $u$. Finally, the choices of $f$ and $B(v)$, for each original vertex $v$, are such that all the original vertices of $G'$ are under illusion, and this can only be corrected if at least $k(v)$ neighbours of $v$ are relabelled. From the above observations we obtain the following:

\begin{lemma}\label{lem:ei-tvd}
A solution to the \elilShort problem on $(G',f)$ corresponds to a solution of the \textsc{TVD} problem on $(G,k)$, restricted to the original vertices of $G'$.
\end{lemma}

\section{The problem is hard}
In this section we provide two reductions, both showing hardness under very restricted properties.

The first reduction is from the \textsc{Set Cover} problem. Given a connected bipartite graph~$I=(S\cup U,E)$ the \textsc{Set Cover} problem asks for a smallest \emph{cover}~$C$ of~$U$, i.e., a minimum size subset~$C$ of~$S$ such that every vertex in~$U$ is adjacent to at least one vertex in~$C$. This problem is known to be $\W[2]$-hard when parameterised by the size of $C$~\cite{DowneyF95}.

\begin{figure}[!t]
\centering
  \includegraphics[width=0.7\textwidth]{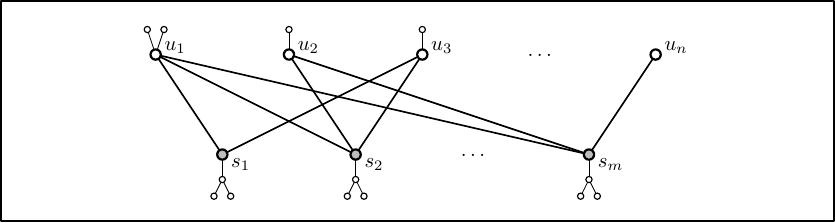}
  \caption{The graph~$G_I$ constructed in the proof of Theorem~\ref{thm:W2-set-cover}. The white vertices have label~$0$, and the grey vertices have label~$1$.}
  \label{fig:set-cover-red}
\end{figure}

\begin{theorem}
    \label{thm:W2-set-cover}
    \elilShort is $\NP$-hard, as well as $\W[2]$-hard when parameterised by the solution size, even if the input graph is bipartite and has bounded diameters.
\end{theorem}
\begin{proof}
Let~$I=(S\cup U,E)$ be an instance of \textsc{Set Cover}. We assume that~$|S|,|U|>0$ and that there is a~$u\in U$ adjacent to all~$s\in S$. We construct the graph~$G_I$ as follows. We start from the graph $I$.
For each~$u\in U$, we add $d_I(u)-1$ leaves attached to~$u$, where $d_I(u)$ is the degree of $u$ in the graph $I$.
For each~$s\in S$ we attach a single vertex, with two leaves, to~$s$.
Hence, the diameter of the graph is at most $6$.
To ease the exposition, we will refer to the vertices of $V(G_I)\cap S$ and $V(G_I)\cap U$ as the vertices of $S$ and $U$ respectively.
We assign~$f(v)=0$ for all~$v\not\in S$ and~$f(v)=1$ for all~$v\in S$.
By the construction of~$G_I$ it follows that~$0$ is the strict majority.
Moreover, only the vertices of $U$ are under illusion.
Furthermore, for each~$u\in U$, it is sufficient that one of its neighbours belonging in $S$ changes its labelling to~$0$ in order for $u$ to no longer be under illusion.
An exemplary visualisation of~$G_I$ can be seen in~\cref{fig:set-cover-red}.

We will show that $I$ has a set cover $C\subseteq S$ of order at most $k$ if and only if there is a solution $f'$ of $G_I$ of size at most $k$.

Let~$C\subset S$ be a covering of~$U$ of size~$k$ in $I$. For every~$s\in C$ we set~$f'(s)=0$ in $G_I$. For all other vertices~$v$ in~$G_I$ we set~$f'(v)=f(v)$. Therefore,~$f$ and~$f'$ differ in exactly~$k$ vertices. Further, since~$C$ is a covering of~$U$, each~$u\in U$ is adjacent to at least one~$s\in C$ whose label was changed to~$0$. Therefore, no vertex in~$G_I$ is under illusion in~$f'$.
\end{proof}

The second hardness reduction is from the \textsc{Planar Monotone} $3$-$\mathsf{SAT}$ problem, a restricted variant of $3$-$\mathsf{SAT}$. In this variant, each clause consists exclusively of either positive or negative literals.
Moreover, the graph admits a straight-line drawing in which all variables lie on a horizontal line. In this representation, every positive (negative) clause is positioned in the upper (lower) half-plane.
This problem is well known to be $\NP$-complete~\cite{deBergKhosravi}.
Similar to the proof idea of Theorem 3 in \cite{DARMANN202145}, we can further assume that each literal appears at most three times.

\begin{figure}[!t]
\centering
  \includegraphics[width=0.7\textwidth]{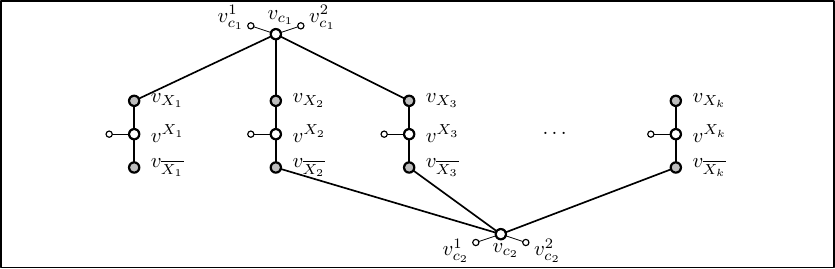}
  \caption{The graph~$G_\phi$ constructed in the proof of Theorem~\ref{thm:NP-3sat}. The white vertices have label~$0$, and the grey vertices have label~$1$. The clauses are $c_1 = (X_1 \lor X_2 \lor X_3)$ and $c_2 = (\overline{X_2} \lor \overline{X_3} \lor \overline{X_k})$.}
  \label{fig:3-sat-red}
\end{figure}

\begin{theorem}
    \label{thm:NP-3sat}
    \elilShort is $\NP$-hard, even if the input graph $G$ is restricted to be bipartite, planar and of maximum degree~$5$.
\end{theorem}

\begin{proof}
Given a formula $\phi$ that is an instance of \textsc{Planar Monotone} $3$-$\mathsf{SAT}$ where each literal appears in at most three clauses, we construct an instance~$G_\phi$ of \elilShort. For every variable $X$, we add the variable-vertices $v_X, v_{\overline{X}}$. For every clause $c \in \phi$, we add a clause-vertex~$v_c$. For every clause-vertex $v_c$, we add the edge $v_cv_X$ ($v_cv_{\overline{X}}$ resp.) for every literal $X$ ($\overline{X}$ resp.) that appears in $c$.
Further, for every pair of literal vertices~$v_X, v_{\overline{X}}$, we add a new vertex~$v^X$, referred to as the \emph{controller} of $X$, as well as the edges $v^Xv_X$ and $v^Xv_{\overline{X}}$.
For every clause-vertex~$v_c$, we add two leaves $v_c^1$ and $v_c^2$.
Finally, we add a leaf attached to the controller~$v^X$ of every literal~$X$.
Let $G_\phi$ be the resulting graph.
It is straightforward to see that $G_\phi$ is a planar bipartite graph.
Moreover, observe that, in $G_\phi$, the degree of every vertex is at most 5.
Thus, $G_\phi$ is bipartite, planar and of maximum degree~$5$.
For every vertex $v=v_X$ or $v=v_{\overline{X}}$, for every literal $X$, we assign $f(v) = 1$.
For every other vertex $v$ of $G_\phi$, we assign~$f(v) = 0$.
Notice that~$0$ is the strict majority, and that the vertices under illusion are precisely the clause-vertices~$v_c$ and the controllers~$v^X$.
Each of them requires at least one of their neighbours to change its label from~$1$ to~$0$ in order to not be under illusion.
A visualisation of $G_\phi$ can be seen in Fig~\ref{fig:3-sat-red}.

Let~$k$ be the number of variables in $\phi$.
We will show that there exists a solution~$f'$ for~$G_\phi$ of size at most~$k$ if and only if~$\phi$ is satisfiable.

If~$\phi$ is satisfiable, there exists a truth assignment $t$ from the variables to $\{$\textsc{True, False}$\}$. For every variable $X$, we set $f'(v_X)=0,f'(v_{\overline{X}})=1$ if $t(X)=\textsc{True}$, and $f'(v_X)=1$, $f'(v_{\overline{X}})=0$ if $t(X)=\textsc{False}$. For all other vertices~$v$ in~$G_\phi$, we set $f'(v) = f(v) = 0$. In this way, $f$ and $f'$ differ on precisely~$k$ vertices.
For every~$X$, the controller~$v^X$ is not under illusion in~$f'$ since the label of precisely one of~$v_X$, $v_{\overline{X}}$ was changed. For every clause~$c$, we know that~$c$ was satisfied by~$t$. Hence, there is at least one variable-vertex that neighbours~$v_c$ whose label was changed and~$c$ is not under illusion in~$f'$. These were the only vertices under illusion in~$f$, thus~$f'$ is a solution for~$G_\phi$.

Conversely, assume there exists a solution~$f'$ for $G_\phi$ of size exactly $k$. Since no vertex is under illusion under~$f'$, for every variable $X$, at least one neighbour~$v_X$ or $v_{\overline{X}}$ of~$v^X$ was relabelled. Actually, it is exactly one of the vertices $v_X$ or $v_{\overline{X}}$ that was relabelled, since there are exactly $k$ variables. Since~$f'$ is a valid solution, every clause is adjacent to at least one vertex that was relabelled. Thus, the truth assignment that sets $X$ to \textsc{True} (\textsc{False} resp.) for each variable $X$ such that~$f'(v_X) = 0$ ($f'(v_{\overline{X}}) = 0$ resp.), yields a satisfying truth assignment for~$\phi$.
\end{proof}

\section{Structural parameters}

\begin{theorem}\label{thm:FPT-vertex cover}
\elilShort is solvable in $\FPT$ time parameterised by the vertex cover number of the input graph.
\end{theorem}

\begin{proof}

Let $G=(V,E)$ be the input graph and $f$ be the initial labelling of $G$. Moreover, let $\mathsf{vc}$ be the vertex cover number of $G$ and let $U\subseteq V$ be a vertex cover of $G$ of minimum size. Recall that $I=V\setminus U$ is an independent set of $G$. Since $|U|\leq \mathsf{vc}$, we may guess an optimal labelling $f'|_U$ of $G$ such that no vertex of $I$ is under illusion by $f'|_U$ (there are at most $2^{\mathsf{vc}}$ such labellings). The labelling $f'|_U$ is optimal in the sense that it differs from $f$ on a minimum number of vertices. All that remains to be done is to extend $f'|_U$ into an optimal solution $f'$ of $G$ by making sure that $f'$ induces no illusion on the vertices of $U$.

To achieve this, we first arrange the vertices of $I$ into sets according to their neighbourhoods in $U$. In particular, we partition $I$ into the sets $I_1,\dots,I_p$, for $p\leq 2^\mathsf{vc}$, such that for every $i\in[p]$, the vertices of $I_i$ are twins, \emph{i.e.}, they have the same neighbourhood. Formally, for every $u,v\in I$, we have that $u\in I_i$ and $v\in I_i$, for some $i\in[p]$, if and only if $N_G(u)=N_G(v)$. Then we compute the exact number of vertices of $I_i$, for each $i\in[p]$, whose label must be changed in order for the resulting labelling $f'$ to be a solution of \elilShort by modelling this problem as an ILP on bounded number of variables.

\medskip

\hrule

\smallskip

\noindent Variables

\begin{tabular}{lll}
$x_i$ & $i\in [p]$  & number of vertices of $I_i$ labelled $1$ by $f'$\\
\end{tabular}

\smallskip

\noindent Constants

\begin{tabular}{lll}
$a(i)$ & $i\in [p]$  &

\begin{tabular}{@{}c@{}}
number of vertices of $I_i$ \\ labelled $1$ by $f$
\end{tabular}
\end{tabular}

\setlength\extrarowheight{-5pt}
\begin{tabular}{lll}
$n(u)$ & $u\in U$ &
\begin{tabular}{@{}c@{}}
number of neighbours of $u$ \\ in $U$ labelled $1$ by $f'|_U$
\end{tabular}\\

\\

$d(u)$ & $u\in U$ &
\begin{tabular}{@{}c@{}}
degree of $u$ in $G$
\end{tabular}\\

\\

$i(u,i)$ & $u\in U, i\in [p]$ &
\begin{tabular}{@{}c@{}}
set to $0$ if $u\notin N(v), \forall v\in I_i,$ \\ and to $1$ otherwise
\end{tabular}
\end{tabular}

\smallskip

\noindent Objective
\begin{align}
\max \sum_{i \in [p]} x_i
\end{align}

\smallskip

\noindent Constraints
\begin{align}
n(u)+i(u,i)\cdot x_i\leq \frac{d(u)}{2} && \forall i\in[p],u\in U\label{model:no-illusion}\\
x_i\leq a(i) && i\in[p]\label{model:no-more-swaps}
\end{align}

\hrule

\medskip

The variable $x_i$ represents the number of vertices of $I_i$ that were labelled $1$ by $f$ and whose label will remain unchanged in $f'$. Constraint~\ref{model:no-more-swaps} makes sure that there are not more vertices in $I_i$ labelled $1$ by $f'$ than they were by $f$. Constraint~\ref{model:no-illusion} ensures that no vertex $u\in U$ is under illusion by $f'$. Since the model has $p\leq 2^{\mathsf{vc}}$ variables, we can and obtain the $x_i$s' in FPT time, parameterised by $\mathsf{vc}$ (by running the Lenstra algorithm~\cite{Len83}). Finally, note that the $x_i$s' are enough to compute a solution $f'$. Indeed, it is sufficient to change the label of any $|I_i|-x_i$ vertices of $I_i$, for each $i\in[p]$, from $1$ to $0$ in order to extend $f'|_U$ into $f'$. This is immediate by the definition of the $I_i$s'.
\end{proof}

From \cref{thm:W2-set-cover} we already know that \elilShort is $\W[2]$-hard if parameterised by the solution size.
However, it is in $\XP$ parameterised by the solution size~$k$, since trying all possible solutions of size at most $k$ is bounded by~$\mathcal{O}(|V|^k)$. Looking for other promising parameters, and in view of Corollary~\ref{cor:trees}, one could hope for the existence of an $\FPT$ algorithm for \elilShort parameterised by the treewidth of the input graph. In the following theorem, we show that this is highly unlikely. We provide a reduction from the \textsc{Minimum Maximum Outdegree} (\textsc{MMO} for short) problem. In this problem, we are given a graph $G= (V,E)$, an edge weighting $w\colon E \rightarrow \mathbb{N}$ and a bound $R \in \mathbb{N}$ (both $w$ and $R$ are given in unary). The question is to find an edge orientation $E'$ of $E$, such that for every $v \in V$, the weighted outdegree of $v$ in $G' = (V, E')$ is at most $R$.
It was recently shown that \textsc{MMO} is $\XNLP$-hard when parameterised by the pathwidth of the input graph~\cite{BodlaenderCW22}. This means that \textsc{MMO} is $W[t]$-hard for every~$t \in \mathbb{Z}_{\geq 1}$. It follows from our reduction that:

\begin{theorem}
\label{thm:treewidth_hard}
    \elilShort is $\XNLP$-hard parameterised by the pathwidth of the input graph.
\end{theorem}

\begin{proof}
    Let~$x_1, x_2$ be two vertices with label $1$. A \emph{switch structure} between the vertices $x_1$ and $x_2$ is a vertex~$y$ with label $0$ that is incident to both of them and has an additional leaf attached to $y$ labelled $0$. Notice that~$y$ is under illusion, and in order for this to change at least one of~$x_1, x_2$ must be relabelled.
    A \emph{chessboard} of size~$k$ is a graph with~$4k$ vertices labelled $1$, arranged on a~$2k \times 2$ grid with a switch structure between all vertex pairs having a Manhattan distance of $1$ on the grid. We denote by \emph{grid vertices} the vertices of a chessboard that also belong to the underlying grid.
    Because of the switch structures, for every pair $u,v$ of adjacent grid vertices, either $u$ or $v$ (or both) must be relabelled to eliminate all illusions in the chessboard. Therefore, a chessboard of size~$k$ requires at least~$2k$ label changes. Further,~$2k$ label changes are enough if and only if all relabelled vertices have an even distance to each other on the underlying grid.
    Note that the pathwidth of a chessboard is constant, since traversing the grid from left to right yields a decomposition with bags of constant size.

\begin{figure}[!t]
\centering
  \includegraphics[width=0.7\textwidth]{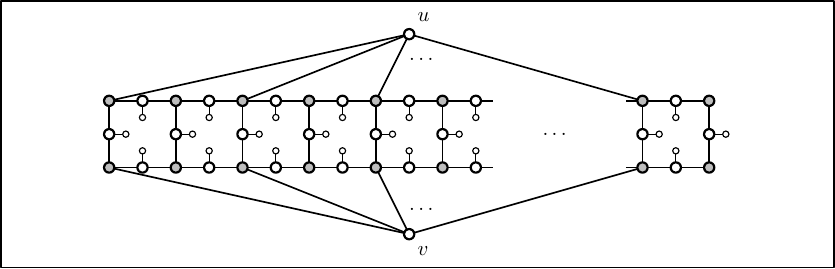}
  \caption{A chessboard structure between vertices $u$ and $v$, used in the proof of Theorem~\ref{thm:treewidth_hard}. The grey vertices have label $1$, and the others label $0$.}
  \label{fig:maxmin-a}
\end{figure}

    Now let~$G=(V,E)$, $w$, $R$ be an instance of \textsc{MMO}.
    Slightly abusing the notation, we denote the weighted degree of~$v$ in~$G$ by $d_G(v)$, for every~$v \in V$. If~$E'$ is an edge orientation of~$E$ and~$G' = (V, E')$ is the corresponding graph, we denote by~$d^-_{G'}(v)$ ($d^+_{G'}(v)$ resp.) the weighted indegree (outdegree resp.) of~$v$ in~$G'$.
    Further, we denote by $W$ the sum of edge weights in $G$.
    Observe that for every edge orientation~$E'$, and for any~$v \in V$, we have that $d^+_{G'}(v) \leq R$ if and only if $d^-_{G'}(v) \geq d_G(v) - R$.

    We construct a new graph~$G^*$ with vertex labels in the following way.
    First, we add all vertices in~$V$ to~$G^*$ and label them with $0$.
    For each~$e = uv\in E$, we replace~$e$ with a chessboard $\mathcal{C}_e$ of size~$w(e)$. For $i\in[2]$ and $j\in[2k]$, let $z_{i,j}$ be the grid vertex of $\mathcal{C}_e$ that lies on the $i^{th}$ row and $j^{th}$ column in the underlying grid of $\mathcal{C}_e$. Then, we connect~$u$ to $z_{1,1}, z_{1,3}, \dots, z_{1,2k-1}$ and $v$ to $z_{2,1}, z_{2,3}, \dots, z_{2,2k-1}$. This construction is illustrated in Fig.~\ref{fig:maxmin-a}.
    For every~$v\in V(G)$, we add leaves attached to~$v$ such that the demand of $v$, i.e., difference between $|\{u\in N(v):f(u)=1\}|$ and $|\{u\in N(v):f(u)=0\}|$, in~$G^*$ becomes precisely $d_G(v) - R$. This can be achieved by adding leaves labelled $1$ ($0$ resp.) as long as the demand of~$v$ is smaller (larger resp.) than $d_G(v) - R$. Since every added leaf can change the demand by at most one, and two leaves of the same type change the demand of~$v$ by at least two, this procedure can be applied to set an arbitrary demand for~$v$.
    Adding the chessboards and the leaves increases the pathwidth of the resulting graph only by a constant. Further, the number of vertices we added is linear in~$W + R$. Since~$w$ and~$R$ were given in unary, the size of~$G^*$ is polynomial in~$(G, w, R)$.

    We will now show that there exists a solution for \elilShort in~$G^*$, labelled as described above, of size~$2W$ if and only if there exists a solution for \textsc{MMO} for $G=(V,E), w, R$.

\begin{figure}[!t]
\centering
  \includegraphics[width=0.7\textwidth]{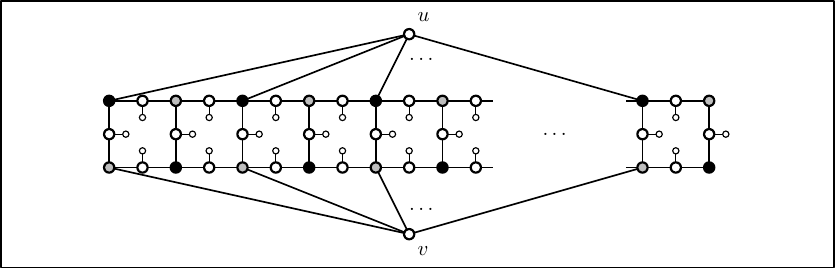}
  \caption{The same structure as in Fig.~\ref{fig:maxmin-a}. The grey vertices were relabelled with label $0$. This corresponds to the edge $uv\in E(G)$ being directed from $u$ to $v$ in $G^*$. The only other valid relabelling with size $2 w(u, v)$ is the reversed relabelling.}
   \label{fig:maxmin-b}
\end{figure}

    Assume there exists a solution $E'$ for \textsc{MMO} in $G=(V,E),$ $ w, R$. For every~$(u,v) \in E'$, we relabel the vertices of $\mathcal{C}_{uv}$ that are adjacent to~$v$. Then, we also relabel the grid vertices of $\mathcal{C}(uv)$ that are not adjacent to $u$ and were not relabelled in the previous step. 
    In this manner, we relabelled precisely~$2w(u,v)$ of the chessboard vertices that replaced the edge~$uv$. In total, we relabel~$2W$ vertices. See Fig.~\ref{fig:maxmin-b} for such a relabelling.
    We claim that after this relabelling, no vertex in~$G^*$ remains under illusion. We have already shown that the relabelling of the chessboards satisfies every internal vertex of the chessboard. Further, all leaves that were added to adjust the demand for the original vertices have only one neighbour with the label $0$, therefore they were not under illusion from the start.
    Thus, the only vertices left to verify, are those that came from~$V(G)$. Let us consider such a vertex $v$. Since~$E'$ is a valid solution of \textsc{MMO}, we have that $d^-_{G'}(v) \geq d_G(v) - R$. Therefore, $v$ has at least $d(v) - R$ relabelled vertices as neighbours, and these relabelled neighbours are precisely those chessboard vertices that correspond to incoming edges of~$v$ in~$G'$. Since this was the demand of~$v$ in~$G^*$, it follows that~$v$ is not under illusion in the relabelled graph.

    Next, assume that there exists a solution~$f^*$ of~$G^*$ of size~$2W$. For every edge~$e\in E$, we have that $\mathcal{C}_e$ requires at least~$2w(e)$ relabellings. Since the total sum of sizes among all included chessboards is~$W$, this means that every chessboard is labelled optimally. Therefore, for every edge $e = uv \in E$, either all neighbours of~$u$ in $\mathcal{C}_{uv}$, or all neighbours of~$v$ in $\mathcal{C}_{uv}$, are relabelled. We add~$(u,v)$ to~$E'$ if all neighbours of~$v$ are relabelled, otherwise, we add~$(v,u)$ to~$E'$. Then, we set~$G' = (V, E')$. Now, for every~$v\in V(G')$, we have that $d^-_{G'}(v)$ is equal to the number of relabelled neighbours of~$v$, which is at least $d_G(v) - R$.
    Equivalently, $d^+_{G'}(v) \leq R$ for every~$v\in V(G')$, and the edge orientation $E'$ is a valid solution for~$G, w, R$.
\end{proof}

\section{A \textsf{PTAS} for planar graphs}

In this section, we use the classic layering technique introduced by \cite{Baker94} for designing approximation algorithms on planar graphs. On a high level, we break the input graph into layers to solve the problem optimally in each layer.
Then we take a union of these solutions to return a feasible solution for the original input.
Our algorithm computes this solution on several \emph{layered decompositions} and returns the minimum among them.
In order to solve the problem in each layer, we run the $\XP$ algorithm (parameterised by the treewidth) that solves the following generalisation of the \elilShort problem, which follows from the upcoming Theorem~\ref{thm:fpt-tw+sol}.

\defproblem{\selil}
{A graph $G=(V,E)$, a set $S\subseteq V$, a labelling $f:V\rightarrow \{0,1\}$ and an integer $k\geq1$.}
{Is there a labelling $f':V\rightarrow \{0,1\}$ such that $f'$ induces no illusion on the vertices of $S$ and $|\{v\in V:f(v)\neq f'(v)\}|\leq k$.}

Before we proceed with the next theorem, allow us to introduce some additional notation. We denote the set $\{0,\dots k\}$ by $[k]$. When considering the \selil problem, let $\mu: S \rightarrow [M]$ be a function that is equal to the minimum number of vertices adjacent to $v$ that must be relabelled such that $v$ is not under illusion any more, for every vertex $v\in S$. We will also extend this function such that $\mu(v)=0$ for every $v\in V\setminus S$. We say that $\mu$ is the \emph{demand function} and denote $M$ as the \emph{maximum demand}.

\begin{theorem}
    \label{thm:fpt-tw+sol}
    There exists an \FPT algorithm for \emph{\selil} parameterised by the treewidth $\operatorname{tw}(G)$ of the input graph and the maximum demand $M$, with running time ${M} ^{\mathcal{O}(\operatorname{tw}(G))} \cdot n ^{\mathcal{O}(1)}$.
\end{theorem}
\begin{proof}

    Let $(G,S,f)$ be an instance of \selil and let $\mathcal{T} = (T, \{X_t\}_{t \in V(T)})$ be a nice tree decomposition of $G$, with $T$ being rooted at a leaf $r \in V(T)$. For some $t \in V(T)$, we define the subgraph $G_t$ of $G$ to be the graph induced by the union of bags contained in the subtree of $T$ rooted at $t$. For instance, the induced graph $G_r$ with respect to the subtree rooted at the root node $r$ is precisely $G$.
    We use dynamic programming on $\mathcal{T}$ to find the minimum number of vertices in $V(G)$ whose label is required to change, such that no vertex in $S$ remains under illusion. 
    Before we describe how the dynamic step is performed on each type of node in the tree decomposition, we introduce some notation. 
    We use vectors $D_t \in [M]^{|X_t|}$ to express the remaining demand for vertices in $X_t$. In particular, for every $v\in X_t\cap S$, the value of $D_t[v]$ is the number of vertices that still need to be relabelled in the neighbourhood of $v$ in order to bring $v$ out of illusion. Also, for every $v\in X_t\setminus S$, we set $D_t[v]=0$.
    We also define some operations on these vectors. We denote by $D_t + \delta_v$ the operation where for each $w \in (N(v) \cap X_t)\cap S$, the entry of $D_t[w]$ is increased by one and for each $w\in (N(v)\cap X_t)\setminus S$, the entry of $D_t[w]$ remains unchanged. Furthermore, we write $\restr{D_t}{D_i}$ for the restriction of the vector $D_t$ to $D_i$. That is, $\restr{D_t}{D_i}$ is a vector whose entries for the vertices $X_t \cap X_i$ is the same as $D_t$.
    Finally, we use $D^{v=a}$ to represent a vector with an entry $D[v]$ equal to $a$.

    We define $W[t, D_t, U]$ to be the minimum number of relabellings such that:
    \begin{enumerate}
        \item the labels of vertices in $U \subseteq X_t$ are relabelled,
        \item the labels of vertices in $X_t \setminus U$ are not relabelled,
        \item\label{property:no-illusion} no vertex of $(V(G_t)\cap S)\setminus X_t$ is under illusion and
        \item\label{property:remaining-demand} vertices $v \in X_t$ have a remaining demand of at most $D_t[v]$.
    \end{enumerate}
    From this definition, it follows that the minimum number of relabellings required is exactly $W[r, \emptyset, \emptyset]$. We now proceed to describe the recursive formulas for the different node types of $\mathcal{T}$.

        \noindent\textbf{Leaf node.} A leaf node $\ell \neq r$ corresponds to an empty graph. Thus $W[\ell, D_\ell, U] = 0$, as no vertex can be under illusion. Therefore, leaf nodes serve as our base case.

        \noindent\textbf{Introduce node.} Let $v$ be the vertex that has been introduced in node $t$ and let node $i$ be the only child of $t$.
        If $D_t[v] < \mu(v) - |N(v) \cap U |$, we set $W[t, D_t, U] = \infty$ since this violates property \ref{property:remaining-demand} as by definition of $\mathcal{T}$ the neighbourhood of $v$ that has been introduced so far must be contained in $X_t$.
        Otherwise, we distinguish between two cases. If $v \in U$, then relabelling $v$ decreases the remaining demand for all of its adjacent vertices. Thus, $v$ relaxes the remaining demand vector $D_t$ by one for all vertices in the neighbourhood of $v$ introduced so far. We can therefore calculate $W[t, D_t, U]$ with the help of the child $i$ and the relaxed vector:
        $W[t, D_t, U] = W[i, \restr{D_t}{D_i} + \delta_v, U \setminus \{v\}] + 1.
        $
        If $v \notin U$, we can look up the required number of relabelling in the child node, since in this case $U \subseteq X_i$ by the definition of a nice tree decomposition:
        $
            W[t, D_t, U] = W[i, \restr{D_t}{D_i}, U].
        $
        To process a specific node of type introduced, we need to consider each possible vector $D_t$ and each possible subset $U \subseteq X_t$. For a given $D_t$ and $U$ we can look up the solution in polynomial time in the size of the bag $X_t$. This yields an overall runtime of $\mathcal{O}((M + 1)^{|X_t|}\cdot 2^{|X_t|} \cdot |X_t|^c) \subseteq \mathcal{O}(\operatorname{tw}(G)^c \cdot (M +1)^{\operatorname{tw}(G)}\cdot 2^{\operatorname{tw}(G)})$, where $c$ is a constant, to process a node of type introduce.

        \noindent\textbf{Forget node.} Let $v$ be the vertex that has been forgotten at node $t$ with child node $i$. Since $v$ is removed from $X_t$, we need to guarantee that property \ref{property:no-illusion} holds, i.e., the remaining demand of $v$ must be $0$. If $v\notin S$, then we are done, as in this case $\mu(v)=D_t[v]=0$ by definition. So, we may assume that $v\in S$. For a vector $D_t^{v=0}$, the vertex $v$ can either be in the set of vertices $U$ that are relabelled or not. We take the minimum of all these cases:
        $
            W[t, D_t, U] = \min\{W[i, D_t^{v=0}, U], W[i, D_t^{v=0}, U \cup \{v\}\}.
        $
        The processing time of a node of type forget is analogous to the processing time of a node of type introduce, resulting in the same runtime of $\mathcal{O}(\operatorname{tw}(G)^c \cdot (M +1)^{\operatorname{tw}(G)}\cdot 2^{\operatorname{tw}(G)})$, where $c$ is a constant.

        \noindent\textbf{Join node.} Let $t$ be a join node with children $i$ and $j$. Node $t$ merges two previously disjoint subgraphs $G_i$ and $G_j$. Consider a vertex $v \in X_t = X_i = X_j$. The vector $D_t[v]$ describes the remaining demand of $v$ by the property \ref{property:remaining-demand}. Thus, at least $\mu(v) - D_t[v]$ vertices have been relabelled in the neighbourhood of $v$. These relabelled vertices in the neighbourhood of $v$ can either be already forgotten in $G_i$ or $G_j$ or still remain in $U \subseteq X_t$. Therefore,
        $
            \mu(v) - D_t[v] = (\mu(v) - D_i[v]) + (\mu(v) - D_j[v]) - |N(v) \cap U|$ from which follows that
            $D_i[v] + D_j[v] = D_t[v] + \mu(v) - |N(v) \cap U|,
        $
        where we need to subtract $|N(v) \cap U|$ since these vertices are accounted for in $\mu(v) - D_i[v]$ and $\mu(v) - D_j[v]$.
        To determine the minimum number of relabellings $W[t, D_t, U]$, we must identify combinations $C_v$ of values $(a_i, b_j)$ satisfying $a_i + b_j = D_t[v] + \mu(v) - |N(v) \cap U|$ for each vertex $v \in X_t$. These pairs represent distinct valid pairings of vectors $D_i$ and $D_j$ for the specific entry $v$.
        For every valid combination in $C_v$, we can pair it with all other valid combinations $C_u$ of vertices $u \in X_t \setminus \{v\}$ to generate vector pairs $(D_i', D_j') \in C(D_t)$ that are valid for all vertices $v \in X_t.$ Using $C_t$, we can then calculate $W[t, D_t, U]$ with the following equation:
        $
            W[t, D_t, U] = \min_{(D_i', D_j') \in C(D_t)}\{W[i, D_i', U] + W[j, D_j', U] - |U|\}.
        $
        To compute a join node, we need to calculate $C_v$ for each $v \in X_t$ for a given $D_t$ and $U$. This can be done in $\mathcal{O}(M)$ time. To construct a single vector pair in $C(D_t)$, we need $\mathcal{O}(|X_t|)$ time. In total, we have $|C(D_t)| \in \mathcal{O}((M + 1)^{|X_t|})$ many valid vectors that need to be checked, yielding a processing time of $\mathcal{O}((M + 1)^{|X_t|}\cdot 2^{|X_t|}\cdot(M + 1)^{|X_t|} \cdot M \cdot |X_t|) \subseteq \mathcal{O}((M + 1)^{2\operatorname{tw}(G)+1}\cdot 2^{\operatorname{tw}(G)} \cdot \operatorname{tw}(G)^c)$, for some constant~$c$.

    The number of nodes in $V(T)$ is in $\mathcal{O}(\operatorname{tw}(G) \cdot n)$, where $n = |V(G)|$. So the runtime of our algorithm is $\mathcal{O}((M + 1)^{2\operatorname{tw}(G)+1}\cdot 2^{\operatorname{tw}(G)} \cdot \operatorname{tw}(G)^c \cdot n)$, for some constant $c$.\end{proof}

Since the maximum demand $M$ is a lower bound for the solution size $k$, which is upper bounded by the number of vertices $n$ in the input graph, we obtain the following.


\begin{corollary}
\label{cor:treewidth-xp}
    There exists an $\FPT$ algorithm for \selil parameterised by the treewidth $\operatorname{tw}$ of the input graph plus the solution size. This implies an $\XP$ algorithm for the same problem parameterised just by $\operatorname{tw}$.
\end{corollary}

We are now ready to present the $\PTAS$ for planar graphs. A description of our algorithm is given in Algorithm~\ref{alg:approx}.

\begin{algorithm}
\caption{$\PTAS$ on planar graphs}
\label{alg:approx}
\KwIn{$G$, $f$, $\eps$}
$k = 4/\eps$\;
Find the outerplanar layers of the vertices\;
\For{$i = 0, \ldots, k-1$}{
	Find the components $G^i_1, G^i_2, \ldots,$ of $G$, where $G^i_j$ contains vertices on the layers from $i+j(k+1)$ to $i+(j+1)(k+1)+1$\;
		\For{$j = 0, 1,\ldots, $}{
			Compute $A_j^{i}$, the solution of \selil on $G_j^{i}$ setting $S = \overline{G}_j^i$ where $\overline{G}_j^i$ contains vertices on the layers from $i+j(k+1)+1$ to $i+(j+1)(k+1)$\;
           	$A^{i} = \cup_j A_j^i$\;
		}
}
Let $A$ be the minimum solution among $\{A^0, A^1, \ldots, A^{k-1} \}$\;
\KwRet $A$
\end{algorithm}

\begin{theorem}
\label{thm:ptas}
For any given $\eps >0$, there exists an approximation algorithm which computes the solution of the \elilShort problem on a planar graph in $\mathcal{O}(n^\frac{1}{\eps})$ time and computes a solution which is $(1+\eps)$-times the optimum.
\end{theorem}
\begin{proof}
For given $\eps > 0$, let $k$ be the nearest positive integer of $\frac{4}{\eps}$.
Consider a planar embedding of the input graph $G = (V, E)$.
We assign a layer to each vertex according to its depth from the outerface, starting from layer zero for the vertices on the outerface.
Clearly, the vertices of layer $\ell$ are on the outerface of the residual graph obtained after deleting the vertices of the layers $\{0,1, 2, \dots,\ell -1\}$.
Clearly, each $G^i_j$ $G^i_j$, as defined in Algorithm~\ref{alg:approx}, is at most a $k$-outerplanar graph. That is, for each vertex $v$ of $G^i_j$, there exists a sequence of at most $k$ consecutive layers of $G$ such that $v$ belongs in the $k^{th}$ such layer.
It is known that every $k$-outerplanar graph has treewidth bounded by $3k -1$~\cite{Bodlaender98}.
It follows from \cref{cor:treewidth-xp} that there exists an algorithm which runs in $n^{\mathcal{O}(k)}\cdot n^{\mathcal{O}(1)}$ time and solves \selil optimally on $G^i_j$ for each $i, j$ considering $S = \overline{G}_j^i$ as defined in Algorithm~\ref{alg:approx}.
For every $i$, we join the solutions we computed for the $G^i_j$'s.
We return the solution which is minimum among all $i$'s.
Hence, the running time of this algorithm is $\mathcal{O}(n^{\frac{1}{\eps}})$.

Let $A^i_j$ be the solution of the problem on $G^i_j$ for each $i,j$ using \cref{cor:treewidth-xp}.
Let $\underline{G}^i_j \subset G^i_j$ be the induced subgraph on the vertices of layers $i+j(k+1)+2$ to $i+(j+1)(k+1)-1$.
Let $\overline{A}^i_j = A^i_j \cap \underline{G}^i_j$.
Let $S$ be the minimum set of vertices that must be relabelled to eliminate all the illusions in $G$.
Let $\overline{S}_j^i = S \cap \underline{G}_j^i$.
Clearly, $\overline{A}^i_j = \overline{S}_j^i$.
Let $V^i_j \subset V$ be the vertices on layers $i+j(k+1)$ and $i+j(k+1)+1$.
Let $S_j^i = S \cap V_j^i$.
Let $S^i = \bigcup_j S_j^i$.
Observe that $\overline{S}^i_j$ and $S^i$ are disjoint and $\sum_i^k |S^i| = 2|S|$.
Hence, there exists one $i^*$ such that $|S^{i^*}| \leq 2|S|/k$.
Let $A = \bigcup_j A^{i^*}_j$ be the union of the solution of the $G^i_j$'s.
Then, $A$ is a potential candidate set of vertices for relabelling for the graph $G$, as for a fixed $i$ we get that the $G^i_j$'s cover the entire graph $G$.
Let, $A_j = A \cap V^{i^*}_j$ be the output of the algorithm restricted on the vertices on layers $i^*+j(k+1)$ and $i^*+j(k+1)+1$.

We claim that for every $j$, we have that $|{A}_j| \leq 2|S_j^{i^*}|$.
Indeed, the set $S_j^{i^*}$ contains the subset of the optimum solution restricted to the vertices on layers $i^*+j(k+1)$ and $i^*+j(k+1)+1$.
Hence, $S^{i^*}_j$ can be seen as the union of two sets $\bar{S}^{i^*}_{j}$ and $\bar{S}^{i^*}_{j'}$ where $\bar{S}^{i^*}_{j}$ contains the vertices on layers $i^*+j(k+1)$ and $\bar{S}^{i^*}_{j'}$ contains the vertices on layers $i^*+j(k+1)+1$.
Moreover, the set ${A}_j$ can be seen as a union of four sets; $A^{i^*}_{j}$, $A^{i^*}_{j'}$, $A^*_{j}$ and $A^*_{j'}$ where
\begin{itemize}
    \item $A^{i^*}_{j}$ contains vertices from the layer $i^*+ j(k+1)$ that appear in the solution of the piece $G^{i^*}_{j-1}$; and
    \item $A^{i^*}_{j'}$ contains vertices from the layer $i^*+ j(k+1)+1$ that appear in the solution of the piece $G^{i^*}_{j}$; and
    \item $A^{*}_{j}$ contains the vertices from the layer $i^*+ j(k+1)$ that appear in the solution of the piece $G^{i^*}_{j}$; and
    \item $A^{*}_{j'}$ contains the vertices from the layer $i^*+j(k+1)+1$ that appear in the solution of the piece $G^{i^*}_{j-1}$.
\end{itemize}
Now, neighbours of a vertex at some layer $i$ can only be present in layers $i-1, i$ and $i+1$.
Hence, vertices of $A^{i^*}_{j}$ are essential for the solution of $G^{i^*}_{j-1}$ and vertices of $A^{*}_{j}$ are essential for the solution of $G^{i^*}_{j}$.
Thus, $|A^{*}_{j} \cup A^{i^*}_{j}| \leq 2|\bar{S}^{i^*}_{j}|$.
Following a similar argument, we can show that $|A^{*}_{j'} \cup A^{i^*}_{j'}| \leq 2|\bar{S}^{i^*}_{j'}|$.
Together, these two inequalities imply that $|{A}_j| \leq 2|S_j^{i^*}|$.

To finish the proof, it suffices to set $k$ to be $4/\eps$.
As discussed above, a solution using our algorithm is a feasible solution for the problem.
Consider the solution $A$ which is the minimum among different choices of $i$.
It holds that $ |A| = \sum_j |A_j^{i^*}| + \sum_j|A_j| \leq |S| + 4|S|/k  = (1+\eps)|S|$.
Hence the theorem.\end{proof}

\section{Conclusion}
In this paper, we initiated the algorithmic study of the \elilShort problem. The main takeaway message is that the problem is computationally hard. Thus some compromises have to be made in order to solve it efficiently. In this work, we decided to focus on solutions whose correctness is theoretically guaranteed. This leaves the more heuristic-oriented approach largely untouched (and highly motivated).


\section*{Acknowledgement}
The authors would like to thank the organisers of Homonolo 2023 for providing an excellent atmosphere for collaboration and research. They would also like to thank an anonymous referee from MFCS 2024 for pointing out the related problem of total vector domination. Abhiruk Lahiri would like to thank Rajesh Chitnis for stimulating discussion on this problem before the inception of this project and pointing out a minor error in the manuscript later on.
Foivos Fioravantes was supported by the International Mobility of Researchers MSCA-F-CZ-III at CTU in Prague, $\text{CZ}.02.01.01/00/22\_010/0008601$ Programme Johannes Amos Comenius.
Abhiruk Lahiri's research was supported by the Strategic Research Fund (SFF), Heinrich Heine University D\"{u}sseldorf.
Lluís Sabater was partially supported by Charles University project UNCE 24/SCI/008, and by the projects 22-22997S and 25-17221S of GA \v{C}R.

%
%
%

\setcitestyle{numbers}
\bibliographystyle{plainnat}
\bibliography{references}

\end{document}